\documentclass{entcs}
\usepackage{entcsmacro}
\usepackage{cite}
\usepackage{amsmath}
\usepackage{amssymb,stmaryrd,mathtools,xspace,url,bbm}
\usepackage[matrix,arrow,cmtip,curve]{xy}
\SelectTips{cm}{}\SilentMatrices\CompileMatrices
\usepackage{tikz}
\usetikzlibrary{arrows,automata,shapes}

\newcommand*\cat[1]{\text{\textup{\textsf{#1}}}}
\newcommand*\HDA{\cat{HDA}}
\newcommand*\HDP{\cat{HDP}}
\newcommand*\HDT{\cat{HDT}}
\newcommand*\HDAt{\cat{HDO}}
\newcommand*\HDAh{\HDA_\textup{\textsf{h}}}
\newcommand*\HDPh{\HDP_\textup{\textsf{h}}}
\newcommand*\HDTh{\HDT_\textup{\textsf{h}}}
\newcommand*\LHDA{\cat{L}\HDA}
\newcommand*\pCub{\cat{pCub}}
\newcommand*\LHDAh{\LHDA_\textup{\textsf{h}}}

\newcommand{\Nat}{\mathbbm{N}}

\DeclareMathOperator{\id}{id}
\newcommand*\tto[1]{\xrightarrow{#1}}
\newcommand*\tfrom[1]{\xleftarrow{#1}}
\newcommand*\from{\leftarrow}
\newcommand*\bang{\mathord{!}}

\newcommand*\ie{\textit{i.e.}\xspace}
\newcommand*\eg{\textit{e.g.}\xspace}
\newcommand*\cf{\textit{cf.}\xspace}
\newcommand*\etal{\textit{et.al.}\xspace}

\def\lastname{Fahrenberg, Legay}

\begin{document}

\begin{frontmatter}
  \title%
  {History-Preserving Bisimilarity \\ for Higher-Dimensional Automata \\
    via Open Maps}
  
  \author{Uli Fahrenberg}
  \author{Axel Legay}
  \address{INRIA/IRISA, Campus de Beaulieu, 35042 Rennes CEDEX, France}
  
  \begin{abstract}
    We show that history-preserving bisimilarity for higher-dimensional
    automata has a simple characterization directly in terms of
    higher-dimensional transitions.  This implies that it is decidable
    for finite higher-dimensional automata.  To arrive at our
    characterization, we apply the open-maps framework of Joyal, Nielsen
    and Winskel in the category of unfoldings of precubical sets.
  \end{abstract}
\end{frontmatter}

\section{Introduction}

The dominant notion for behavioral equivalence of processes is
\emph{bisimulation} as introduced by Park~\cite{DBLP:conf/tcs/Park81}
and Milner~\cite{book/Milner89}.  It is compelling because it enjoys
good algebraic properties, admits several easy characterizations using
modal logics, fixed points, or game theory, and generally has low
computational complexity.

Bisimulation, or rather its underlying semantic model of
\emph{transition systems}, applies to a setting in which concurrency of
actions is the same as non-deterministic interleaving; using CCS
notation~\cite{book/Milner89}, $a| b= a. b+ b. a$.  For some
applications however, a distinction between these two is necessary,
which has led to development of so-called \emph{non-interleaving} or
\emph{truly concurrent} models such as Petri nets~\cite{book/Petri62},
event structures~\cite{DBLP:journals/tcs/NielsenPW81}, asynchronous
transition systems~\cite{Bednarczyk87-async,DBLP:journals/cj/Shields85}
and others; see~\cite{WinskelN95-Models} for a survey.

One of the most popular notions of equivalence for non-interleaving
systems is \emph{history-preserving bisimilarity} (or
\emph{hp-bisimilarity} for short).  It was introduced independently by
Degano, De Nicola and Montanari in~\cite{DeganoNM89} and by Rabinovich
and Trakhtenbrot~\cite{journals/fundinf/RabinovichT88} and then for
event structures by van~Glabbeek and Goltz
in~\cite{DBLP:conf/mfcs/GlabbeekG89} and for Petri nets by Best
\etal~in~\cite{DBLP:journals/acta/BestDKP91}.  One reason for its
popularity is that it is a congruence under action
refinement~\cite{DBLP:conf/mfcs/GlabbeekG89,DBLP:journals/acta/BestDKP91},
another its good decidability properties: it has been shown to be
decidable for safe Petri nets by Montanari and
Pistore~\cite{DBLP:conf/stacs/MontanariP97}.  As a contrast, its cousin
\emph{hereditary} hp-bisimilarity is shown undecidable for $1$-safe
Petri nets by Jurdzi{\'n}ski, Nielsen and Srba
in~\cite{DBLP:journals/iandc/JurdzinskiNS03}.

\emph{Higher-dimensional automata} (or \emph{HDA}) is another
non-interleaving formalism for reasoning about behavior of concurrent
systems.  Introduced by Pratt~\cite{Pratt91-geometry} and
van~Glabbeek~\cite{Glabbeek91-hda} in 1991 for the purpose of a
\emph{geometric} interpretation to the theory of concurrency, it has
since been shown by van~Glab\-beek~\cite{DBLP:journals/tcs/Glabbeek06}
that HDA provide a generalization (up to hp-bisimilarity) to ``the main
models of concurrency proposed in the
literature''~\cite{DBLP:journals/tcs/Glabbeek06}, including the ones
mentioned above.  Hence HDA are useful as a tool for comparing and
relating different models, and also as a modeling formalism by
themselves.

HDA are geometric in the sense that they are very similar to the
\emph{simplicial complexes} used in algebraic topology, and research on
HDA has drawn on a lot of tools and methods from geometry and algebraic
topology such as
homotopy~\cite{DBLP:journals/tcs/FajstrupRG06,DBLP:journals/mscs/Gaucher00},
homology~\cite{DBLP:conf/concur/GoubaultJ92,journals/ctopgd/Gaucher02},
and model
categories~\cite{journals/tac/Gaucher11,journals/tac/Gaucher09}, see
also the survey~\cite{DBLP:journals/mscs/Goubault00a}.

In this paper we give a geometric interpretation to hp-bisimilarity for
HDA, using the open-maps approach introduced by Joyal, Nielsen and
Winskel in~\cite{DBLP:journals/iandc/JoyalNW96} and results from a
previous paper~\cite{Fahrenberg05-hda} by the first author.  Using this
interpretation, we show that hp-bisimilarity for HDA has a
characterization directly in terms of (higher-dimensional)
\emph{transitions} of the HDA, rather than in terms of runs as \eg~for
Petri nets~\cite{DBLP:conf/mfcs/FroschleH99}.

Our results imply \emph{decidability} of hp-bisimilarity for finite
HDA.  They also put hp-bisimilarity firmly into the open-maps framework
of~\cite{DBLP:journals/iandc/JoyalNW96} and tighten the connections
between bisimilarity and weak topological
\emph{fibrations}~\cite{AdamekHRT02-weak,KurzR05-weak}.

Due to lack of space, we have had to confer all proofs of this paper to
a separate appendix.

\section{Higher-Dimensional Automata}
\label{sec:hda}

As a formalism for concurrent behavior, HDA have the specific feature
that they can express all higher-order dependencies between events in a
concurrent system.  Like for transition systems, they consist of states
and transitions which are labeled with events.  Now if two transitions
from a state, with labels $a$ and $b$ for example, are independent, then
this is expressed by the existence of a \emph{two-dimensional}
transition with label $ab$.  Fig.~\ref{fig:independence} shows two
examples; on the left, transitions $a$ and $b$ are independent, on the
right, they can merely be executed in any order.  Hence for HDA, as
indeed for any formalism employing the so-called \emph{true concurrency}
paradigm, the algebraic law $a| b= a. b+ b. a$ does \emph{not} hold;
concurrency is not the same as interleaving.

The above considerations can equally be applied to sets of more than two
events: if three events $a$, $b$, $c$ are independent, then this is
expressed using a three-dimensional transition labeled $abc$.  Hence
this is different from mutual pairwise independence (expressed by
transitions $ab$, $ac$, $bc$), a distinction which cannot be made in
formalisms such as asynchronous transition
systems~\cite{Bednarczyk87-async,DBLP:journals/cj/Shields85} or
transition systems with independence~\cite{WinskelN95-Models} which only
consider binary independence relations.

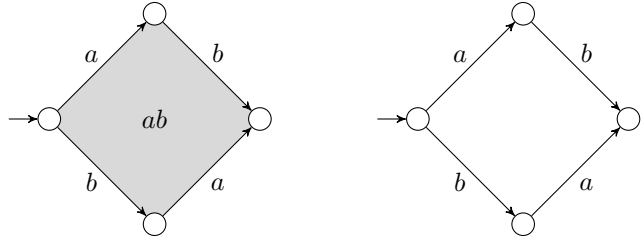
\begin{figure}
  \centering
  \begin{tikzpicture}[->,>=stealth',auto,scale=.7]
    \tikzstyle{every node}=[font=\footnotesize,initial text=]
    \tikzstyle{every state}=[fill=white,shape=circle,inner
    sep=.5mm,minimum size=3mm]
    \begin{scope}
      \path[fill=black!15] (0,0) to (2,2) to (4,0) to (2,-2);
      \node[state,initial] (s0) at (0,0) {};
      \node[state] (s1) at (2,2) {};
      \node[state] (s2) at (2,-2) {};
      \node[state] (s3) at (4,0) {};
      \path (s0) edge (s1);
      \path (s0) edge (s2);
      \path (s1) edge (s3);
      \path (s2) edge (s3);
      \node at (.8,1.2) {$a$};
      \node at (.8,-1.2) {$b$};
      \node at (3.2,1.25) {$b$};
      \node at (3.2,-1.25) {$a$};
      \node at (2,0) {$ab$};
    \end{scope}
    \begin{scope}[xshift=7cm]
      \node[state,initial] (s0) at (0,0) {};
      \node[state] (s1) at (2,2) {};
      \node[state] (s2) at (2,-2) {};
      \node[state] (s3) at (4,0) {};
      \path (s0) edge (s1);
      \path (s0) edge (s2);
      \path (s1) edge (s3);
      \path (s2) edge (s3);
      \node at (.8,1.2) {$a$};
      \node at (.8,-1.2) {$b$};
      \node at (3.2,1.25) {$b$};
      \node at (3.2,-1.25) {$a$};
    \end{scope}
  \end{tikzpicture}
  \caption{%
    \label{fig:independence}
    HDA for the CCS expressions $a| b$ (left) and $a. b+ b. a$ (right).
    In the left HDA, the square is filled in by a two-dimensional
    transition labeled $ab$, signifying independence of events $a$ and
    $b$.  On the right, $a$ and $b$ are not independent.%
  }
\end{figure}

An unlabeled HDA is essentially a pointed precubical set as defined
below.  For labeled HDA, one can pass to an arrow category; this is what
we shall do in Section~\ref{sec:labels}.  Until then, we concentrate on
the unlabeled case.

A \emph{precubical set} is a graded set $X= \{ X_n\}_{ n\in \Nat}$
together with mappings $\delta_k^\nu:X_n\to X_{ n- 1}$, $k\in\{ 1,\dots,
n\}$, $\nu\in\{ 0, 1\}$, satisfying the \emph{precubical identity}
\begin{equation}
  \label{eq:pcub}
  \delta_k^\nu \delta_\ell^\mu= \delta_{ \ell- 1}^\mu
  \delta_k^\nu \qquad( k< \ell)\,.
\end{equation}
The mappings $\delta_k^\nu$ are called \emph{face maps}, and elements of
$X_n$ are called \emph{$n$-cubes}.  As above, we shall usually omit the
extra subscript $(n)$ in the face maps.  Faces $\delta_k^0 x$ of an
element $x\in X$ are to be thought of as \emph{lower faces}, $\delta_k^1
x$ as \emph{upper faces}.  The precubical identity expresses the fact
that $( n- 1)$-faces of an $n$-cube meet in common $( n- 2)$-faces, see
Fig.~\ref{fig:2cubefaces} for an example of a $2$-cube and its faces.

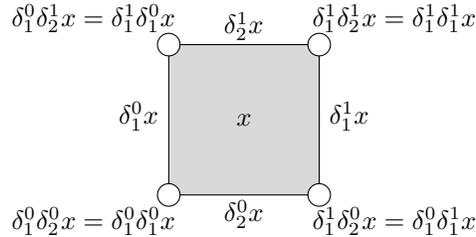
\begin{figure}[b]
  \centering
  \begin{tikzpicture}[->,>=stealth',auto,scale=1]
    \tikzstyle{every node}=[font=\footnotesize]
    \tikzstyle{every state}=[fill=white,shape=circle,inner
    sep=.5mm,minimum size=3mm]
    \path[fill=black!15] (0,0) to (2,0) to (2,2) to (0,2) to (0,0);
    \path (0,0) edge (2,0);
    \path (2,0) edge (2,2);
    \path (0,0) edge (0,2);
    \path (0,2) edge (2,2);
    \node[state] at (0,0) {};
    \node[state] at (2,0) {};
    \node[state] at (0,2) {};
    \node[state] at (2,2) {};
    \node at (1,1) {$x$};
    \node at (-.4,1.05) {$\delta_1^0 x$};
    \node at (2.4,1.05) {$\delta_1^1 x$};
    \node at (1,-.25) {$\delta_2^0 x$};
    \node at (1,2.25) {$\delta_2^1 x$};
    \node at (-1,-.35) {$\delta_1^0 \delta_2^0 x= \delta_1^0
      \delta_1^0 x$};
    \node at (-1,2.35) {$\delta_1^0 \delta_2^1 x= \delta_1^1
      \delta_1^0 x$};
    \node at (3,-.35) {$\delta_1^1 \delta_2^0 x= \delta_1^0
      \delta_1^1 x$};
    \node at (3,2.35) {$\delta_1^1 \delta_2^1 x= \delta_1^1
      \delta_1^1 x$};
  \end{tikzpicture}
  \caption{%
    \label{fig:2cubefaces}
    A $2$-cube $x$ with its four faces $\delta_1^0 x$, $\delta_1^1 x$,
    $\delta_2^0 x$, $\delta_2^1 x$ and four corners.
  }
\end{figure}

\emph{Morphisms} $f: X\to Y$ of precubical sets are graded mappings
$f=\{ f_n: X_n\to Y_n\}_{ n\in \Nat}$ which commute with the face maps:
$\delta_k^\nu\circ f_n= f_{ n- 1}\circ \delta_k^\nu$ for all $n\in
\Nat$, $k\in\{ 1,\dots, n\}$, $\nu\in\{ 0, 1\}$.  This defines a
category $\pCub$ of precubical sets and morphisms.

A \emph{pointed} precubical set is a precubical set $X$ with a
specified $0$-cube $i\in X_0$, and a pointed morphism is one which
respects the point.  This defines a category which is isomorphic to
the comma category $*\downarrow \pCub$, where $*\in \pCub$ is the
precubical set with one $0$-cube and no other $n$-cubes.  Note that
$*$ is \emph{not} terminal in $\pCub$ (instead, the terminal object is
the infinite-dimensional precubical set with one cube in every
dimension).

\begin{definition}
  \label{defi:hda}
  The category of \emph{higher-dimensional automata} is the comma
  category $\HDA= *\downarrow \pCub$, with objects pointed
  precubical sets and morphisms commutative diagrams
  \begin{equation*}
    \xymatrix@C=1.3em@R=1.1em{%
      & {*} \ar[dl] \ar[dr] \\ X \ar[rr]_f && Y\,.
    }
  \end{equation*}
\end{definition}

Hence a one-dimensional HDA is a transition system; indeed, the category
of transition systems~\cite{WinskelN95-Models} is isomorphic to the full
subcategory of $\HDA$ spanned by the one-dimensional objects.  Similarly
one can show~\cite{Goubault02-cmcim} that the category of asynchronous
transition systems is isomorphic to the full subcategory of $\HDA$
spanned by the (at most) two-dimensional objects.  The category $\HDA$
as defined above was used in~\cite{Fahrenberg05-hda} to provide a
categorical framework (in the spirit of~\cite{WinskelN95-Models}) for
parallel composition of HDA.  In this article we also introduced a
notion of bisimilarity which we will review in the next section.

\section{Path Objects, Open Maps and Bisimilarity}

With the purpose of introducing bisimilarity via \emph{open maps} in the
sense of~\cite{DBLP:journals/iandc/JoyalNW96}, we identify here a
subcategory of $\HDA$ consisting of path objects and path-extending
morphisms.  We say that a precubical set $X$ is a \emph{precubical path
  object} if there is a (necessarily unique) sequence $( x_1,\dots,
x_m)$ of elements in $X$ such that $x_i\ne x_j$ for $i\ne j$,
\begin{itemize}
\item for each $x\in X$ there is $j\in\{ 1,\dots, m\}$ for which
  $\smash{ x= \delta_{ k_1}^{ \nu_1}\cdots \delta_{ k_p}^{ \nu_p} x_j}$
  for some indices $\nu_1,\dots, \nu_p$ and a \emph{unique} sequence
  $k_1<\dots< k_p$, and
\item for each $j= 1,\dots, m- 1$, there is $k\in \Nat$ for which $x_j=
  \delta_k^0 x_{ j+ 1}$ or $x_{ j+ 1}= \delta_k^1 x_j$.
\end{itemize}

Note that precubical path objects are \emph{non-selflinked} in the sense
of~\cite{DBLP:journals/tcs/FajstrupRG06}.  If $X$ and $Y$ are precubical
path objects with representations $( x_1,\dots, x_m)$, $( y_1,\dots,
y_p)$, then a morphism $f: X\to Y$ is called a \emph{cube path
  extension} if $x_j= y_j$ for all $j= 1,\dots, m$ (hence $m\le p$).

\begin{definition}
  The category $\HDP$ of \emph{higher-dimensional paths} is the
  subcategory of $\HDA$ which as objects has pointed precubical paths,
  and whose morphisms are generated by isomorphisms and pointed cube
  path extensions.
\end{definition}

A \emph{cube path} in a precubical set $X$ is a morphism $P\to X$ from a
precubical path object $P$.  In elementary terms, this is a sequence $(
x_1,\dots, x_m)$ of elements of $X$ such that for each $j= 1,\dots, m-
1$, there is $k\in \Nat$ for which $x_j= \delta_k^0 x_{ j+ 1}$ (start of
new part of a computation) or $x_{ j+ 1}= \delta_k^1 x_j$ (end of a
computation part).  We show an example of a cube path in
Fig.~\ref{fig:cubepath}.

A cube path in a HDA $i: *\to X$ is \emph{pointed} if $x_1= i$, hence if
it is a pointed morphism $P\to X$ from a higher-dimensional path $P$.
We will say that a cube path $( x_1,\dots, x_m)$ is \emph{from} $x_1$
\emph{to} $x_m$, and that a cube $x\in X$ in a HDA $X$ is
\emph{reachable} if there is a pointed cube path to $x$ in $X$.

\begin{figure}
  \centering
  \begin{tikzpicture}[->,>=stealth',auto,scale=.9]
    \tikzstyle{every node}=[font=\footnotesize,initial text=]
    \tikzstyle{every state}=[fill=white,shape=circle,inner
    sep=.5mm,minimum size=3mm]
    \path[fill=black!15] (2,0) to (4,0) to (4,2) to (2,2);
    \node[state,initial] (i) at (0,0) {};
    \node[state] (x) at (2,0) {};
    \node[state] (z) at (4,2) {};
    \path (i) edge (x); 
    \path (x) edge (4,0); 
    \path (4,0) edge (z); 
    \path (z) edge (6,2); 
    \node at (0,-.35) {$i$};
    \node at (1,-.25) {$a$};
    \node at (2,-.4) {$x$};
    \node at (3,-.2) {$b$};
    \node at (3,1) {$bc$};
    \node at (4.2,.97) {$c$};
    \node at (4,2.3) {$z$};
    \node at (5,2.2) {$d$};
  \end{tikzpicture}
  \caption{%
    \label{fig:cubepath}
    Graphical representation of the two-dimensional cube path $( i, a,
    x, b, bc, c, z, d)$.  Its computational interpretation is that $a$
    is executed first, then execution of $b$ starts, and while $b$ is
    running, $c$ starts to execute.  After this, $b$ finishes, then $c$,
    and then execution of $d$ is started.  Note that the computation is
    partial, as $d$ does not finish.
  }
\end{figure}
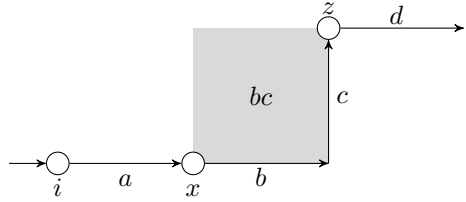

Cube paths can be \emph{concatenated} if the end of one is compatible
with the beginning of the other: If $\rho=( x_1,\dots, x_m)$ and
$\sigma=( y_1,\dots, y_p)$ are cube paths with $y_1= \delta_k^1 x_m$ or
$x_m= \delta_k^0 y_1$ for some $k$, then their \emph{concatenation} is
the cube path $\rho* \sigma=( x_1,\dots, x_m, y_1,\dots, y_p)$.  We say
that $\rho$ is a \emph{prefix} of $\chi$ and write $\rho\sqsubseteq
\chi$ if there is a cube path $\rho$ for which $\chi=\rho* \sigma$.

\begin{definition}
  \label{defi:open_hda}
  A pointed morphism $f: X\to Y$ in $\HDA$ is an \emph{open map} if it
  has the right lifting property with respect to $\HDP$, \ie~if it is
  the case that there is a lift $r$ in any commutative diagram as below,
  for morphisms $g: P\to Q\in \HDP$, $p: P\to X, q: Q\to Y\in \HDA$:
  \begin{equation*}
    \xymatrix{%
      P \ar[r]^p \ar[d]_g & X \ar[d]^f \\ Q \ar[r]_q \ar@{.>}[ur]|r &
      Y
    }
  \end{equation*}
  HDA $X$, $Y$ are \emph{bisimilar} if there is $Z\in \HDA$ and a span
  of open maps $X\from Z\to Y$ in $\HDA$.
\end{definition}

It follows straight from the definition that composites of open maps are
again open.  By the next lemma, morphisms are open precisely when they
have a zig-zag property similar to the one
of~\cite{DBLP:journals/iandc/JoyalNW96}.

\begin{lemma}
  \label{lem:open}
  For a morphism $f: X\to Y\in \HDA$, the following are equivalent:
  \begin{enumerate}
  \item\label{enu:open.box} $f$ is open;
  \item\label{enu:open.onestep} for any reachable $x_1\in X$ and any
    $y_2\in Y$ with $f( x_1)= \delta_k^0 y_2$ for some $k$, there is
    $x_2\in X$ for which $x_1= \delta_k^0 x_2$ and $y_2= f( x_2)$;
  \item\label{enu:open.cubepath} for any reachable $x_1\in X$ and any
    cube path $( y_1,\dots, y_m)$ in $Y$ with $y_1= f( x_1)$, there is a
    cube path $( x_1,\dots, x_m)$ in $X$ for which $y_j= f( x_j)$ for
    all $j= 1,\dots, m$.
  \end{enumerate}
\end{lemma}

\begin{theorem}
  \label{thm:bisim}
  For HDA $i: *\to X$, $j: * \to Y$, the following are equivalent:
  \begin{enumerate}
  \item\label{enu:bisim.box} $X$ and $Y$ are bisimilar;
  \item\label{enu:bisim.onestep} there exists a precubical subset
    $R\subseteq X\times Y$ for which $( i, j)\in R$, and such that for
    all reachable $x_1\in X$, $y_1\in Y$ with $( x_1, y_1)\in R$,
    \begin{itemize}
    \item for any $x_2\in X$ for which $x_1= \delta_k^0 x_2$ for some
      $k$, there exists $y_2\in Y$ for which $y_1= \delta_k^0 y_2$ and
      $( x_2, y_2)\in R$,
    \item for any $y_2\in Y$ for which $y_1= \delta_k^0 y_2$ for some
      $k$, there exists $x_2\in X$ for which $x_1= \delta_k^0 x_2$ and
      $( x_2, y_2)\in R$;
    \end{itemize}
  \item\label{enu:bisim.cubepath} there exists a precubical subset
    $R\subseteq X\times Y$ for which $( i, j)\in R$, and such that for
    all reachable $x_1\in X$, $y_1\in Y$ with $( x_1, y_1)\in R$,
    \begin{itemize}
    \item for any cube path $( x_1,\dots, x_m)$ in $X$, there exists a
      cube path $( y_1,\dots, y_m)$ in $Y$ with $( x_p, y_p)\in R$ for
      all $p= 1,\dots, m$,
    \item for any cube path $( y_1,\dots, y_m)$ in $Y$, there exists a
      cube path $( x_1,\dots, x_m)$ in $X$ with $( x_p, y_p)\in R$ for
      all $p= 1,\dots, m$.
    \end{itemize}
  \end{enumerate}
\end{theorem}

Note that the requirement that $R$ be a precubical subset, in
items~\eqref{enu:bisim.onestep} and~\eqref{enu:bisim.cubepath} above, is
equivalent to saying that whenever $( x, y)\in R$, then also $(
\delta_k^\nu x, \delta_k^\nu y)\in R$ for any $k$ and $\nu\in\{ 0, 1\}$.

\section{Homotopies and Unfoldings}
\label{sec:hdp}

In order to reason about hp-bisimilarity, we need to introduce in which
cases different cube paths are equivalent due to independence of
actions.  Following~\cite{DBLP:journals/tcs/Glabbeek06}, we model this
equivalence by a combinatorial version of \emph{homotopy} which is an
extension of the equivalence defining \emph{Mazurkiewicz
  traces}~\cite{Mazurkiewicz77}.

We say that cube paths $( x_1,\dots, x_m)$, $( y_1,\dots, y_m)$ are
\emph{adjacent} if $x_1= y_1$, $x_m= y_m$, there is precisely one index
$p\in\{ 1,\dots, m\}$ at which $x_p\ne y_p$, and
\begin{itemize}
\item $x_{ p- 1}= \delta_k^0 x_p$, $x_p= \delta_\ell^0 x_{ p+ 1}$, $y_{
    p- 1}= \delta_{ \ell- 1}^0 y_p$, and $y_p= \delta_k^0 y_{ p+ 1}$ for
  some $k< \ell$, or vice versa,
\item $x_p= \delta_k^1 x_{ p- 1}$, $x_{ p+ 1}= \delta_\ell^1 x_p$, $y_p=
  \delta_{ \ell- 1}^1 y_{ p- 1}$, and $y_{ p+ 1}= \delta_k^1 y_p$ for
  some $k< \ell$, or vice versa,
\item $x_p= \delta_k^0 \delta_\ell^1 y_p$, $y_{ p- 1}= \delta_k^0 y_p$,
  and $y_{ p+ 1}= \delta_\ell^1 y_p$ for some $k< \ell$, or vice versa,
  or
\item $x_p= \delta_k^1 \delta_\ell^0 y_p$, $y_{ p- 1}= \delta_\ell^0
  y_p$, and $y_{ p+ 1}= \delta_k^1 y_p$ for some $k< \ell$, or vice
  versa.
\end{itemize}

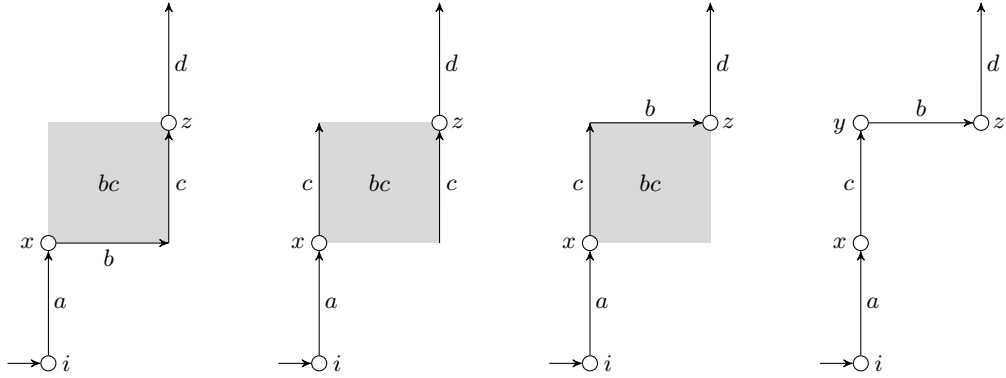
\begin{figure}
  \centering
  \begin{tikzpicture}[->,>=stealth',auto,scale=.8]
    \tikzstyle{every node}=[font=\small,initial text=]
    \tikzstyle{every state}=[fill=white,shape=circle,inner
    sep=.5mm,minimum size=2mm]
    \begin{scope}
      \path[fill=black!15] (2,0) to (4,0) to (4,2) to (2,2);
      \node[state,initial] (i) at (2,-2) {};
      \node[state] (x) at (2,0) {};
      \node[state] (z) at (4,2) {};
      \path (i) edge (x); 
      \path (x) edge (4,0); 
      \path (4,0) edge (z); 
      \path (z) edge (4,4); 
      \node at (2.3,-2) {$i$};
      \node at (1.65,0) {$x$};
      \node at (2.2,-1) {$a$};
      \node at (3,-.25) {$b$};
      \node at (3,1) {$bc$};
      \node at (4.2,.97) {$c$};
      \node at (4.3,2) {$z$};
      \node at (4.2,3) {$d$};
    \end{scope}
    \begin{scope}[xshift=4.5cm]
      \path[fill=black!15] (2,0) to (4,0) to (4,2) to (2,2);
      \node[state,initial] (i) at (2,-2) {};
      \node[state] (x) at (2,0) {};
      \node[state] (z) at (4,2) {};
      \path (i) edge (x); 
      \path (x) edge (2,2);
      \path (4,0) edge (z); 
      \path (z) edge (4,4); 
      \node at (2.3,-2) {$i$};
      \node at (1.65,0) {$x$};
      \node at (2.2,-1) {$a$};
      \node at (1.8,.97) {$c$};
      \node at (3,1) {$bc$};
      \node at (4.2,.97) {$c$};
      \node at (4.3,2) {$z$};
      \node at (4.2,3) {$d$};
    \end{scope}
    \begin{scope}[xshift=9cm]
      \path[fill=black!15] (2,0) to (4,0) to (4,2) to (2,2);
      \node[state,initial] (i) at (2,-2) {};
      \node[state] (x) at (2,0) {};
      \node[state] (z) at (4,2) {};
      \path (i) edge (x); 
      \path (x) edge (2,2);
      \path (2,2) edge (z);
      \path (z) edge (4,4); 
      \node at (2.3,-2) {$i$};
      \node at (1.65,0) {$x$};
      \node at (2.2,-1) {$a$};
      \node at (1.8,.97) {$c$};
      \node at (3,2.25) {$b$};
      \node at (3,1) {$bc$};
      \node at (4.3,2) {$z$};
      \node at (4.2,3) {$d$};
    \end{scope}
    \begin{scope}[xshift=13.5cm]
      \node[state,initial] (i) at (2,-2) {};
      \node[state] (x) at (2,0) {};
      \node[state] (y) at (2,2) {};
      \node[state] (z) at (4,2) {};
      \path (i) edge (x); 
      \path (x) edge (y);
      \path (y) edge (z);
      \path (z) edge (4,4); 
      \node at (2.3,-2) {$i$};
      \node at (1.65,0) {$x$};
      \node at (1.65,1.95) {$y$};
      \node at (2.2,-1) {$a$};
      \node at (1.8,.97) {$c$};
      \node at (3,2.25) {$b$};
      \node at (4.3,2) {$z$};
      \node at (4.2,3) {$d$};
    \end{scope}
  \end{tikzpicture}
  \caption{%
    \label{fig:homcubepath}
    Graphical representation of the cube path homotopy $( i, a, x, b, bc, c, z,
    d)\sim$ $( i, a, x, c, bc, c, z, d)\sim( i, a, x, c, bc, b, z, d)\sim(
    i, a, x, c, y, b, z, d)$.  }
\end{figure}

\emph{Homotopy} of cube paths is the reflexive, transitive closure of
the adjacency relation.  We denote homotopy of cube paths using the
symbol $\sim$, and the homotopy class of a cube path $( x_1,\dots, x_m)$
is denoted $[ x_1,\dots, x_m]$.  The intuition of adjacency is rather
simple, even though the combinatorics may look complicated, see
Fig.~\ref{fig:homcubepath}.  Note that adjacencies come in two basic
``flavors'': the first two above in which the dimensions of $x_\ell$ and
$y_\ell$ are the same, and the last two in which they differ by $2$.

The following lemma shows that, as expected, cube paths entirely
contained in one cube are homotopic (provided that they share
endpoints).

\begin{lemma}
  \label{lem:hom_fullcube}
  Let $x\in X_n$ in a precubical set $X$ and $( k_1,\dots, k_n)$, $(
  \ell_1,\dots, \ell_n)$ sequences of indices with $k_j, \ell_j\le j$
  for all $j= 1,\dots, n$.  Let $x_j= \delta_{ k_j}^0\cdots \delta_{
    k_n}^0 x$, $y_j= \delta_{ \ell_j}^0\cdots \delta_{ \ell_n}^0 x$.
  Then the cube paths $( x_1,\dots, x_n, x)\sim( y_1,\dots, y_n, x)$.
\end{lemma}

We extend concatenation and prefix to homotopy classes of cube paths
by defining $[ x_1,\dots, x_m]*[ y_1,\dots, y_p]=[ x_1,\dots, x_m,
y_1,\dots, y_p]$ and saying that $\tilde x\sqsubseteq \tilde z$, for
homotopy classes $\tilde x$, $\tilde z$ of cube paths, if there are $(
x_1,\dots, x_m)\in \tilde x$ and $( z_1,\dots, z_q)\in \tilde z$ for
which $( x_1,\dots, x_m)\sqsubseteq( z_1,\dots, z_q)$.  It is easy to
see that concatenation is well-defined, and that $\tilde x\sqsubseteq
\tilde z$ if and only if there is a homotopy class $\tilde y$ for which
$\tilde z= \tilde x* \tilde y$.

Using homotopy classes of cube paths, we can now define the
\emph{unfolding} of a HDA.  Unfoldings of HDA are similar to unfoldings
of transition systems~\cite{WinskelN95-Models} or Petri
nets~\cite{DBLP:journals/tcs/NielsenPW81,DBLP:conf/fsttcs/HaymanW08},
but also to \emph{universal covering spaces} in algebraic topology.  The
intention is that the unfolding of a HDA captures all its computations,
up to homotopy.

We say that a HDA $X$ is a \emph{higher-dimensional tree} if it holds
that for any $x\in X$, there is precisely one homotopy class of pointed
cube paths to $x$.  The full subcategory of $\HDA$ spanned by the
higher-dimensional trees is denoted $\HDT$.  Note that any
higher-dimensional path is a higher-dimensional tree; indeed there is an
inclusion $\HDP\hookrightarrow \HDT$.

\begin{definition}
  \label{defi:unfold}
  The \emph{unfolding} of a HDA $i: *\to X$ consists of a HDA $\tilde i:
  *\to \tilde X$ and a pointed \emph{projection} morphism $\pi_X: \tilde
  X\to X$, which are defined as follows:
  \begin{itemize}
  \item $\tilde X_n=\big\{[ x_1,\dots, x_m]\mid( x_1,\dots, x_m)$
    pointed cube path in $X, x_m\in X_n\big\}$; $\tilde i=[ i]$
  \item $\tilde \delta_k^0[ x_1,\dots, x_m]=\big\{ \sigma=( y_1,\dots,
    y_p)\mid y_p= \delta_k^0 x_m, \sigma* x_m\sim( x_1,\dots,x_m)\big\}$
  \item $\tilde \delta_k^1[ x_1,\dots, x_m]=[ x_1,\dots, x_m,
    \delta_k^1 x_m]$
  \item $\pi_X[ x_1,\dots, x_m]= x_m$
  \end{itemize}
\end{definition}

\begin{proposition}
  \label{thm:univ_cov}
  The unfolding $( \tilde X, \pi_X)$ of a HDA $X$ is well-defined, and
  $\tilde X$ is a higher-dimensional tree.  If $X$ itself is a
  higher-dimensional tree, then the projection $\pi_X: \tilde X\to X$ is
  an isomorphism.
\end{proposition}

\begin{lemma}
  \label{lem:cube_path_rep}
  If $X$ is a higher-dimensional automaton and $( \tilde x_1,\dots,
  \tilde x_m)$ is a pointed cube path in $\tilde X$, then $( \pi_X
  \tilde x_1,\dots, \pi_X \tilde x_j)\in \tilde x_j$ for all $j=
  1,\dots, m$.
\end{lemma}

\begin{lemma}
  \label{le:unf_unique_lift}
  For any HDA $X$ there is a unique lift $r$ in any commutative diagram
  as below, for morphisms $g: P\to Q\in \HDP$, $p: P\to \tilde X, q:
  Q\to X\in \HDA$:
  \begin{equation*}
    \xymatrix{%
      P \ar[r]^p \ar[d]_g & \tilde X \ar[d]^{ \pi_X} \\ Q \ar[r]_q
      \ar@{.>}[ur]|r & X
    }
  \end{equation*}
\end{lemma}

\begin{corollary}
  \label{cor:proj_open}
  Projections are open, and any HDA is bisimilar to its
  unfolding. \qed
\end{corollary}

\section{History-Preserving Bisimilarity}
\label{sec:open-hdah}

In this section we recall history-preserving bisimilarity for HDA
from~\cite{DBLP:journals/tcs/Glabbeek06} and show the main result of
this paper: that hp-bisimilarity and the bisimilarity of
Def.~\ref{defi:open_hda} are the same.  To do this, we first need to
introduce \emph{morphisms of homotopy classes of paths} and
\emph{homotopy bisimilarity}.

\begin{definition}
  \label{de:hdah}
  The category of \emph{higher-dimensional automata up to homotopy}
  $\HDAh$ has as objects HDA and as morphisms pointed precubical
  morphisms $f:\tilde X\to \tilde Y$ of unfoldings.
\end{definition}

Hence any morphism $X\to Y$ in $\HDA$ gives, by the unfolding functor,
rise to a morphism $X\to Y$ in $\HDAh$.  The simple example in
Fig.~\ref{fig:hda_vs_hdah} shows that the converse is not the case.  By
restriction to higher-dimensional trees, we get a full subcategory
$\HDTh\hookrightarrow \HDAh$.

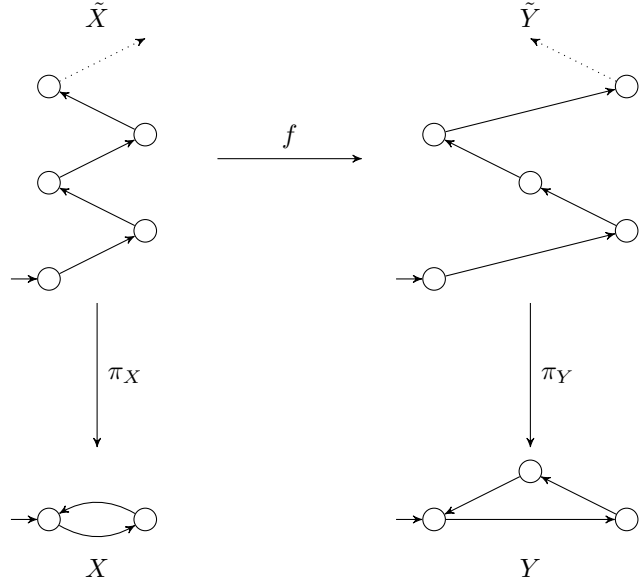
\begin{figure}
  \centering
  \begin{tikzpicture}[->,>=stealth',auto,scale=.64]
    \tikzstyle{every node}=[font=\footnotesize,initial text=]
    \tikzstyle{every state}=[fill=white,shape=circle,inner
    sep=.5mm,minimum size=3mm]
    \begin{scope}
      \node[state,initial] (x) at (0,0) {};
      \node[state] (y) at (2,0) {};
      \path[out=-30,in=-150] (x) edge (y);
      \path[out=-210,in=30] (y) edge (x);
      \node at (1,-1) {$X$};
      \path (1,4.5) edge node {$\pi_X$} (1,1.5);
    \end{scope}
    \begin{scope}[yshift=5cm]
      \node[state,initial] (x0) at (0,0) {};
      \node[state] (y0) at (2,1) {};
      \node[state] (x1) at (0,2) {};
      \node[state] (y1) at (2,3) {};
      \node[state] (x2) at (0,4) {};
      \path (x0) edge (y0);
      \path (y0) edge (x1);
      \path (x1) edge (y1);
      \path (y1) edge (x2);
      \path[dotted] (x2) edge (2,5);
      \node at (1,5.5) {$\tilde X$};
      \path (3.5,2.5) edge node {$f$} (6.5,2.5);
    \end{scope}
    \begin{scope}[xshift=8cm]
      \node[state,initial] (x) at (0,0) {};
      \node[state] (y) at (4,0) {};
      \node[state] (z) at (2,1) {};
      \path (x) edge (y);
      \path (y) edge (z);
      \path (z) edge (x);
      \node at (2,-1) {$Y$};
      \path (2,4.5) edge node {$\pi_Y$} (2,1.5);
    \end{scope}
    \begin{scope}[xshift=8cm,yshift=5cm]
      \node[state,initial] (x0) at (0,0) {};
      \node[state] (y0) at (4,1) {};
      \node[state] (z0) at (2,2) {};
      \node[state] (x1) at (0,3) {};
      \node[state] (y1) at (4,4) {};
      \path (x0) edge (y0);
      \path (y0) edge (z0);
      \path (z0) edge (x1);
      \path (x1) edge (y1);
      \path[dotted] (y1) edge (2,5);
      \node at (2,5.5) {$\tilde Y$};
    \end{scope}
  \end{tikzpicture}
  \caption{%
    \label{fig:hda_vs_hdah}
    Two simple one-dimensional HDA as objects of $\HDA$ and $\HDAh$.  In
    $\HDA$ there is no morphism $X\to Y$, in $\HDAh$ there is precisely
    one morphism $f: X\to Y$.
  }
\end{figure}

\begin{lemma}
  \label{lem:hdt_vs_hdth}
  The natural projection isomorphisms $\pi_X: \tilde X\to X$ for $X\in
  \HDT$ extend to an isomorphism of categories $\HDTh\cong \HDT$.
\end{lemma}

Restricting the above isomorphism to the subcategory $\HDP$ of $\HDT$
allows us to identify a subcategory $\HDPh$ of $\HDTh$ isomorphic to
$\HDP$.

\begin{definition}
  A pointed morphism $f: X\to Y$ in $\HDAh$ is \emph{open} if it has
  the right lifting property with respect to $\HDPh$, \ie~if it is the
  case that there is a lift $r$ in any commutative diagram as below,
  for all morphism $g: P\to Q\in \HDPh$, $p: P\to X, q: Q\to Y\in
  \HDAh$:
  \begin{equation*}
    \xymatrix{%
      P \ar[r]^p \ar[d]_g & X \ar[d]^f \\ Q \ar[r]_q \ar@{.>}[ur]|r &
      Y
    }
  \end{equation*}
  HDA $X$, $Y$ are \emph{homotopy bisimilar} if there is $Z\in \HDAh$
  and a span of open maps $X\from Z\to Y$ in $\HDAh$.
\end{definition}

The connections between open maps in $\HDAh$ and open maps in $\HDA$ are
as follows.

\begin{lemma}
  \label{lem:openh=open}
  \label{le:liftopen}
  A morphism $f: X\to Y$ in $\HDAh$ is open if and only if $f: \tilde
  X\to \tilde Y$ is open as a morphism of $\HDA$.  If $g: X\to Y$ is
  open in $\HDA$, then so is $\tilde g: \tilde X\to \tilde Y$.
\end{lemma}

We also need a lemma on prefixes in unfoldings.

\begin{lemma}
  \label{lem:reach-vs-prefix}
  Let $X$ be a HDA and $\tilde x, \tilde z\in \tilde X$.  Then there is
  a cube path from $\tilde x$ to $\tilde z$ in $\tilde X$ if and only if
  $\tilde x\sqsubseteq \tilde z$.
\end{lemma}

\begin{proposition}
  \label{thm:hombis}
  For HDA $i: *\to X$, $j: * \to Y$, the following are equivalent:
  \begin{enumerate}
  \item\label{enu:hombis.box} $X$ and $Y$ are homotopy bisimilar;
  \item\label{enu:hombis.onestep} there exists a precubical subset
    $R\subseteq \tilde X\times \tilde Y$ with $( \tilde i, \tilde
    j)\in R$, and such that for all $( \tilde x_1, \tilde y_1)\in R$,
    \begin{itemize}
    \item for any $\tilde x_2\in \tilde X$ for which $\tilde x_1=
      \delta_k^0 \tilde x_2$ for some $k$, there exists $\tilde y_2\in
      \tilde Y$ for which $\tilde y_1= \delta_k^0 \tilde y_2$ and $(
      \tilde x_2, \tilde y_2)\in R$,
    \item for any $\tilde y_2\in \tilde Y$ for which $\tilde y_1=
      \delta_k^0 \tilde y_2$ for some $k$, there exists $\tilde x_2\in
      \tilde X$ for which $\tilde x_1= \delta_k^0 \tilde x_2$ and $(
      \tilde x_2, \tilde y_2)\in R$;
    \end{itemize}
  \item\label{enu:hombis.cubepath} there exists a precubical subset
    $R\subseteq \tilde X\times \tilde Y$ with $( \tilde i, \tilde
    j)\in R$, and such that for all $( \tilde x_1, \tilde y_1)\in R$,
    \begin{itemize}
    \item for any cube path $( \tilde x_1,\dots, \tilde x_n)$ in
      $\tilde X$, there exists a cube path $( \tilde y_1,\dots, \tilde
      y_n)$ in $\tilde Y$ with $( \tilde x_p, \tilde y_p)\in R$ for
      all $p= 1,\dots, n$,
    \item for any cube path $( \tilde y_1,\dots, \tilde y_n)$ in
      $\tilde Y$, there exists a cube path $( \tilde x_1,\dots, \tilde
      x_n)$ in $\tilde X$ with $( \tilde x_p, \tilde y_p)\in R$ for
      all $p= 1,\dots, n$;
    \end{itemize}
  \item\label{enu:hombis.prefix} there exists a precubical subset
    $R\subseteq \tilde X\times \tilde Y$ with $( \tilde i, \tilde
    j)\in R$, and such that for all $( \tilde x_1, \tilde y_1)\in R$,
    \begin{itemize}
    \item for any $\tilde x_2\sqsupseteq \tilde x_1$ in $\tilde X$,
      there exists $\tilde y_2\sqsupseteq \tilde y_1$ in $\tilde Y$
      for which $( \tilde x_2, \tilde y_2)\in R$,
    \item for any $\tilde y_2\sqsupseteq \tilde y_1$ in $\tilde Y$,
      there exists $\tilde x_2\sqsupseteq \tilde x_1$ in $\tilde X$
      for which $( \tilde x_2, \tilde y_2)\in R$.
    \end{itemize}
  \end{enumerate}
\end{proposition}

Again, the requirement that $R$ be a precubical subset is equivalent to
saying that whenever $( \tilde x, \tilde y)\in R$, then also $(
\delta_k^\nu \tilde x, \delta_k^\nu \tilde y)\in R$ for any $k$ and
$\nu\in\{ 0, 1\}$.  The next result is what will allow us to relate
hp-bisimilarity and bisimilarity.

\begin{theorem}
  \label{th:hbis=bis}
  HDA $X$, $Y$ are homotopy bisimilar if and only if they are bisimilar.
\end{theorem}

The following is an unlabeled version of hp-bisimilarity
for HDA as defined in~\cite{DBLP:journals/tcs/Glabbeek06}:

\begin{definition}
  \label{defi:histpres}
  HDA $i: *\to X$, $j: *\to Y$ are \emph{history-preserving bisimilar}
  if there exists a relation $R$ between pointed cube paths in $X$ and
  pointed cube paths in $Y$ for which $(( i),( j))\in R$, and such that
  for all $( \rho, \sigma)\in R$,
  \begin{itemize}
  \item for all $\rho'\sim \rho$, there exists $\sigma'\sim \sigma$
    with $( \rho', \sigma')\in R$,
  \item for all $\sigma'\sim \sigma$, there exists $\rho'\sim \rho$
    with $( \rho', \sigma')\in R$,
  \item for all $\rho'\sqsupseteq \rho$, there exists
    $\sigma'\sqsupseteq \sigma$ with $( \rho', \sigma')\in R$,
  \item for all $\sigma'\sqsupseteq \sigma$, there exists
    $\rho'\sqsupseteq \rho$ with $( \rho', \sigma')\in R$.
  \end{itemize}
\end{definition}

We are ready to show the main result of this paper, which together with
Theorem~\ref{th:hbis=bis} gives our characterization for
hp-bisimilarity.

\begin{theorem}
  \label{thm:hombis=histpres}
  HDA $X$, $Y$ are homotopy bisimilar if and only if they are
  history-preserving bisimilar.
\end{theorem}

\begin{corollary}
  \label{co:decidable}
  History-preserving bisimilarity is decidable for finite HDA.
\end{corollary}

\section{Labels}
\label{sec:labels}

We finish this paper by showing how to introduce labels into the above
framework of bisimilarity and homotopy bisimilarity.  Also in the
labeled case, we are able to show that the three notions of
bisimilarity, homotopy bisimilarity and history-preserving bisimilarity
agree.

For labeling HDA, we need a subcategory of $\pCub$ isomorphic to the
category of sets and functions.  Given a finite or countably infinite
set $S=\{ a_1, a_2,\dots\}$, we construct a precubical set $\bang S=\{
\bang S_n\}$ by letting
\begin{equation*}
  \bang S_n=\big\{( a_{ i_1},\dots, a_{ i_n})\mid
  i_k\le i_{ k+ 1}\text{ for all } k= 1,\dots,n- 1\big\}
\end{equation*}
with face maps defined by $\delta_k^\nu( a_{ i_1},\dots, a_{ i_n})=(
a_{ i_1},\dots, a_{ i_{ k- 1}}, a_{ i_{ k+ 1}},\dots, a_{ i_n})$.

\begin{definition}
  The category of \emph{higher-dimensional tori} $\HDAt$ is the full
  subcategory of $\pCub$ generated by the objects $\bang S$.
\end{definition}

As any object in $\HDAt$ has precisely one $0$-cube, the pointed
category $*\downarrow \HDAt$ is isomorphic to $\HDAt$.  It is not
difficult to see that $\HDAt$ is indeed isomorphic to the category of
finite or countably infinite sets and functions,
\cf~\cite{DBLP:journals/entcs/GoubaultM12}.

\begin{definition}
  The category of \emph{labeled higher-dimensional automata} is the
  pointed arrow category $\LHDA= *\downarrow \pCub\to \HDAt$, with
  objects $*\to X\to \bang S$ labeled pointed precubical sets and
  morphisms commutative diagrams
  \begin{equation*}
    \xymatrix@C=1.5em@R=1.2em{%
      & {*} \ar[dl] \ar[dr] \\ X \ar[rr]_f \ar[d] && Y \ar[d] \\ {\bang
        S} \ar[rr]_\sigma && {\bang T}
    }
  \end{equation*}
\end{definition}

\begin{definition}
  A morphism $( f, \id):( *\to X\to \bang S)\to( *\to Y\to \bang S)$ in
  $\LHDA$ is \emph{open} if its component $f$ is open in $\HDA$.
  Labeled HDA $*\to X\to \bang S$, $*\to Y\to \bang S$ are
  \emph{bisimilar} if there is $*\to Z\to \bang S\in \LHDA$ and a span
  of open maps $X\from Z\to Y$ in $\LHDA$.
\end{definition}

Next we establish a correspondence between split
traces~\cite{DBLP:journals/tcs/Glabbeek06} and cube paths in
higher-dimensional tori.  For us, a \emph{split trace} over a finite or
countably infinite set $S$ is a pointed cube path in $\bang S$.  Hence
\eg~a split trace $a^+b^+a^-b^+b^-$ (in the notation
of~\cite{DBLP:journals/tcs/Glabbeek06}) corresponds to the cube path $(
i, a, ab, b, bb, b)$.  Both indicate the start of an $a$ event, followed
by the start of a $b$ event, the end of an $a$ event, the start of a $b$
event, and the end of a $b$ event.  Note that contrary to
ST-traces~\cite{DBLP:journals/tcs/Glabbeek06}, the split trace contains
no information as to which of the two $b$ events has terminated at the
$b^-$.

By definition, a torus $\bang S$ on a finite or countably infinite set
$S=\{ a_1, a_2,\dots\}$ contains all $n$-cubes $( a_{ i_1},\dots, a_{
  i_n})$.  Hence we have the following lemma:

\begin{lemma}
  \label{lem:cube_path_hom_bang}
  Let $( x_1,\dots, x_m)$, $( y_1,\dots, y_m)$ be pointed cube paths in
  $\bang S$ with $x_m= y_m$.  Then $( x_1,\dots, x_m)\sim ( y_1,\dots,
  y_m)$. \qed
\end{lemma}

Homotopy classes of split traces are thus determined by their endpoint
and length:

\begin{corollary}
  The unfolding of a higher-dimensional torus $i: *\to \bang S\in
  \HDAt$ is isomorphic to the pointed precubical set $j: *\to Y$ given
  as follows:
  \begin{itemize}
  \item $Y_n=\{( x, m)\mid x\in \bang S_n, m\ge n, m\equiv n\mod 2\}$,
    $j=( i, 0)$
  \item $\delta_k^0( x, m)=( \delta_k^0 x, m- 1)$, $\delta_k^1( x, m)=(
    \delta_k^1 x, m+ 1)$ \qed
  \end{itemize}
\end{corollary}

The definitions of open maps and bisimilarity in $\HDAh$ can now
easily be extended to the labeled case.  Again, we only need
label-preserving morphisms.

\begin{definition}
  The category of \emph{labeled higher-dimensional automata up to
    homotopy} $\LHDAh$ has as objects labeled HDA $*\to X\to \bang S$
  and as morphisms pairs of precubical morphisms $( f, \sigma):( *\to
  \tilde X\to \bang \tilde S)\to( *\to \tilde Y\to \bang \tilde T)$ of
  unfoldings.
\end{definition}

\begin{definition}
  A morphism $( f, \id):( *\to X\to \bang S)\to( *\to Y\to \bang S)$ in
  $\LHDAh$ is \emph{open} if its component $f$ is open in $\HDAh$.
  Labeled HDA $*\to X\to \bang S$, $*\to Y\to \bang S$ are
  \emph{homotopy bisimilar} if there is $*\to Z\to \bang S\in \LHDAh$
  and a span of open maps $X\from Z\to Y$ in $\LHDAh$.
\end{definition}

The proof of the next theorem is exactly the same as the one for
Theorem~\ref{th:hbis=bis}.

\begin{theorem}
  Labeled HDA $X$, $Y$ are homotopy bisimilar if and only if they are
  bisimilar.  \qed
\end{theorem}

Also for the labeled version, we can now show that homotopy bisimilarity
agrees with history-preserving bisimilarity.  We first recall the
definition from~\cite{DBLP:journals/tcs/Glabbeek06}, where we extend
the labeling morphisms to cube paths by $\lambda( x_1,\dots, x_m)=(
\lambda x_1,\dots, \lambda x_m)$:

\begin{definition}
  Labeled HDA $*\tto i X\tto \lambda \bang S$, $*\tto j Y\tto \mu \bang
  S$ are \emph{history-preserving bisimilar} if there exists a relation
  $R$ between pointed cube paths in $X$ and pointed cube paths in $Y$
  for which $(( i),( j))\in R$, and such that for all $( \rho,
  \sigma)\in R$,
  \begin{itemize}
  \item $\lambda( \rho)= \mu( \sigma)$,
  \item for all $\rho'\sim \rho$, there exists $\sigma'\sim \sigma$
    with $( \rho', \sigma')\in R$,
  \item for all $\sigma'\sim \sigma$, there exists $\rho'\sim \rho$
    with $( \rho', \sigma')\in R$,
  \item for all $\rho'\sqsupseteq \rho$, there exists
    $\sigma'\sqsupseteq \sigma$ with $( \rho', \sigma')\in R$,
  \item for all $\sigma'\sqsupseteq \sigma$, there exists
    $\rho'\sqsupseteq \rho$ with $( \rho', \sigma')\in R$.
  \end{itemize}
\end{definition}

\begin{theorem}
  \label{th:hombis=bistpres-l}
  Labeled HDA $X$, $Y$ are homotopy bisimilar if and only if they are
  history-preserving bisimilar.
\end{theorem}

\section{Conclusion}

We have shown that hp-bisimilarity for HDA can be characterized by spans
of open maps in the category of pointed precubical sets, or equivalently
by a zig-zag relation between cubes in all dimensions.  Aside from
implying decidability of hp-bisimilarity for HDA, and together with the
results of~\cite{DBLP:journals/tcs/Glabbeek06}, this confirms that HDA
is a natural formalism for concurrency: not only does it generalize the
main models for concurrency which people have been working with, but it
also is remarkably simple and natural.

One major question which remains is whether also \emph{hereditary}
hp-bisimilarity can fit into our framework.  Because of its
back-tracking nature, it seems that simple unfoldings of HDA are not the
right tools to use; one should rather consider some form of
back-unfoldings of forward-unfoldings.  Given the undecidability result
of~\cite{DBLP:journals/iandc/JurdzinskiNS03}, it seems doubtful,
however, that any characterization as simple as the one we have for
hp-bisimilarity can be obtained.

Another important question is how HDA relate to other models for
concurrency which are not present in the spectrum presented
in~\cite{DBLP:journals/tcs/Glabbeek06}.  One major such formalism is the
one of \emph{history-dependent automata} which have been introduced by
Montanari and Pistore
in~\cite{DBLP:journals/entcs/MontanariP97,DBLP:conf/stacs/MontanariP97}
and have recently attracted attention in model
learning~\cite{DBLP:conf/concur/AartsHV12,DBLP:conf/pts/AartsJU10}.  We
conjecture that up to hp-bisimilarity, HDA are equivalent to
history-dependent automata.

With regard to the geometric interpretation of HDA as directed
topological spaces, there are two open questions related to the work
laid out in the paper: In~\cite{Fahrenberg05-hda} we show that morphisms
in $\HDA$ are open if and only if their geometric realizations lift
pointed directed paths.  This shows that there are some connections to
weak factorization systems~\cite{AdamekHRT02-weak} here which should be
explored; see~\cite{KurzR05-weak} for a related approach.

In~\cite{Fahrenberg05-thesis} we relate homotopy of cube paths to
directed homotopy of directed paths in the geometric realization.  Based
on this, one should be able to prove that the geometric realization of
the unfolding of a higher-dimensional automaton is the same as the
universal directed covering~\cite{FajstrupR08-convenient} of its
geometric realization and hence that morphisms in $\HDAh$ are open if
and only if their geometric realizations lift dihomotopy classes of
pointed dipaths.

The precise relation of our HDA unfolding to the one for Petri
nets~\cite{DBLP:journals/tcs/NielsenPW81,DBLP:conf/fsttcs/HaymanW08} and
other models for concurrency should also be worked out.  A starting
point for this research could be the work on symmetric event structures
and their relation to presheaf categories
in~\cite{DBLP:conf/lics/StatonW10}.

\renewcommand{\thethm}{A.\arabic{thm}}

\clearpage
\section*{Appendix: Proofs}

\begin{proof*}{Proof of Lemma~\ref{lem:open}.}
  For the implication
  \eqref{enu:open.box}~$\Longrightarrow$~\eqref{enu:open.onestep}, let
  $p: P\to X$ be a pointed cube path with $P$ represented by $(
  p_1,\dots, p_m)$ and $p( p_m)= x_1$.  Let $p_{ m+ 1}$ be a cube of
  dimension one higher than $p_m$, set $p_m= \delta_k^0 p_{ m+ 1}$, and
  let $Q$ be the higher-dimensional path represented by $( p_1,\dots,
  p_m, p_{ m+ 1})$.  Let $g: P\to Q$ be the inclusion, and define $q:
  Q\to Y$ by $q( p_j)= f( p( p_j))$ for $j= 1,\dots, m$ and $q( p_{ m+
    1})= y_2$.  We have a lift $r: Q\to X$ and can set $x_2= r( p_{ m+
    1})$.

  The implication
  \eqref{enu:open.onestep}~$\Longrightarrow$~\eqref{enu:open.cubepath}
  can easily be shown by induction.  The case $y_m= \delta_k^0 y_{ m+
    1}$ follows directly from~\eqref{enu:open.onestep}, and the case
  $y_{ m+ 1}= \delta_k^1 y_m$ is clear by $\delta_k^1 \circ f= f\circ
  \delta_k^1$.

  To finish the proof, we show the implication
  \eqref{enu:open.cubepath}~$\Longrightarrow$~\eqref{enu:open.box}.  Let
  \begin{equation*}
    \xymatrix{%
      P \ar[r]^p \ar[d]_g & X \ar[d]^f \\ Q \ar[r]_q & Y
    }
  \end{equation*}
  be a commutative diagram, with $P$ represented by $( p_1,\dots, p_m)$.
  Up to isomorphism we can assume that $Q$ is represented by $(
  p_1,\dots, p_m, p_{ m+ 1},\dots, p_t)$ and that $g$ is the inclusion.
  The cube $p( p_m)$ is reachable in $X$, and $( q( p_m),\dots, q(
  p_t))$ is a cube path in $Y$ which starts in $q( p_m)= f( p( p_m))$.
  Hence we have a cube path $( x_m,\dots, x_t)$ in $X$ with $x_m= p(
  p_m)$ and $q( p_j)= f( x_j)$ for all $j= m,\dots, t$, and we can
  define a lift $r: Q\to X$ by $r( p_j)= p( p_j)$ for $j= 1,\dots, m$
  and $r( p_j)= x_j$ for $j= m+ 1,\dots, t$.
\end{proof*}

\begin{proof*}{Proof of Theorem~\ref{thm:bisim}.}
  For the implication
  \eqref{enu:bisim.box}~$\Longrightarrow$~\eqref{enu:bisim.onestep}, let
  $X\tfrom f Z\tto g Y$ be a span of open maps and define $R=\{( x,
  y)\in X\times Y\mid \exists z\in Z: x= f( z), y= g( z)\}$.  Then $( i,
  j)\in R$ because $f$ and $g$ are pointed morphisms, and the other
  properties follow by Lemma~\ref{lem:open}.  The implication
  \eqref{enu:bisim.onestep}~$\Longrightarrow$~\eqref{enu:bisim.cubepath}
  can be shown by a simple induction, and for the implication
  \eqref{enu:bisim.cubepath}~$\Longrightarrow$~\eqref{enu:bisim.box},
  the projections give a span $X\tfrom{ \pi_1} R\tto{ \pi_2} Y$ and are
  open by Lemma~\ref{lem:open}.
\end{proof*}

\begin{proof*}{Proof of Lemma~\ref{lem:hom_fullcube}
    ({{\cf~\cite[Ex.~2.15]{Fajstrup05-cubcomp}}}).}
  We can represent a cube path $( x_1,\dots, x_n, x)$ as above by an
  element $( p_1,\dots, p_n)$ of the symmetric group $S_n$ by setting
  $p_n= k_n$ and, working backwards, $p_j=(\{ 1,\dots, n\}\setminus\{
  p_{ j+ 1},\dots, p_n\})[ k_j]$, denoting by this the $k_j$-largest
  element of the set in parentheses.  This introduces a bijection
  between the set of cube paths from the lower left corner of $x$ to $x$
  on the one hand, and elements of $S_n$ on the other hand, and under
  this bijection adjacencies of cube paths are transpositions in $S_n$.
  These generate all of $S_n$, hence all such cube paths are
  homotopic.
\end{proof*}

\begin{proof*}{Proof of Proposition~\ref{thm:univ_cov}.}
  Before proving the proposition, we need an auxiliary notion of
  \emph{fan-shaped} cube path together with a technical lemma.  Say that
  a cube path $( x_1,\dots, x_m)$ in a precubical set $X$, with $x_m\in
  X_n$, is fan-shaped if
  \begin{equation*}
    x_j\in
    \begin{cases}
      X_0 &\text{for } 1\le j\le m- n\text{ odd,} \\
      X_1 &\text{for } 1\le j\le m- n\text{ even,} \\
      X_{ n+ j- m} &\text{for } m- n< j\le m\,.
    \end{cases}
  \end{equation*}
  Hence a fan-shaped cube path is a one-dimensional path up to the point
  where it needs to build up to hit the possibly high-dimensional end
  cube $x_m$.

  \begin{lemma}
    \label{lem:fan}
    Any pointed cube path in a higher-dimensional automaton $i: *\to X$
    is homotopic to a fan-shaped one.
  \end{lemma}

  \begin{proof}
    Let us first introduce some notation: For any pointed cube path $(
    x_1,\dots, x_m)$, let $n_1,\dots, n_m\in \Nat$ be such that $x_j\in
    X_{ n_j}$ (hence $n_j$ is the \emph{dimension} of $x_j$), and let
    $T( x_1,\dots, x_m)= n_1+\cdots+ n_m$.  An easy induction shows that
    $j- n_j$ is odd for all $j$.  Also, $T( x_1,\dots, x_m)\ge \frac12(
    n_m^2+ m- 1)$, with equality if and only if $( x_1,\dots, x_m)$ is
    fan-shaped.

    Next we show that $n_1+\cdots+ n_m\equiv \frac12( n_m^2+ m- 1)\mod
    2$.  By oddity of $j- n_j$ we have $\sum_{ j= 1}^m n_j- \sum_{ j=
      1}^m j\equiv m\mod 2$, and also $\frac12( n_m^2+ m- 1)- \sum_{ j=
      1}^m j= \frac12( n_m^2- m^2- 1)\equiv m\mod 2$, hence the claim
    follows.

    We can now finish the proof by showing how to convert a cube path $(
    x_1,\dots, x_m)$ with $T( x_1,\dots, x_m)> \frac12( n_m^2+ m- 1)$
    into an adjacent cube path $( x_1',\dots, x_m')$ which has $T(
    x_1',\dots, x_m')= T( x_1,\dots, x_m)- 2$, essentially by replacing
    one of its cubes, called $x_\ell$ below, with another one of
    dimension $n_\ell- 2$.

    If $( x_1,\dots, x_m)$ is a cube path which is not fan-shaped, then
    there is an index $\ell\in\{3,\dots, m- 1\}$ for which $n_\ell\ge
    2$, $x_{ \ell- 1}= \delta_{ k_2}^0 x_\ell$ for some $k_2$, and $x_{
      \ell+ 1}= \delta_{ k_3}^1 x_\ell$ for some $k_3$.  Assuming $\ell$
    to be the \emph{least} such index, we must also have $x_{ \ell- 2}=
    \delta_{ k_1}^0 x_{ \ell- 1}$ for some $k_1$.

    Now if $k_2< k_3$, then $\delta_{ k_2}^0 x_{ \ell+ 1}= \delta_{
      k_2}^0 \delta_{ k_3}^1 x_\ell= \delta_{ k_3- 1}^1 \delta_{ k_2}^0
    x_\ell= \delta_{ k_3- 1}^1 x_{ \ell- 1}$ by the precubical
    identity~\eqref{eq:pcub}, hence we can let $( x_1',\dots, x_m')$ be
    the cube path with $x_j'= x_j$ for $j\ne \ell$ and $x_\ell'=
    \delta_{ k_2}^0 x_{ \ell+ 1}$.

    If $k_2> k_3$, then similarly $\delta_{ k_3}^1 x_{ \ell- 1}=
    \delta_{ k_3}^1 \delta_{ k_2}^0 x_\ell= \delta_{ k_2- 1}^0 \delta_{
      k_3}^1 x_\ell= \delta_{ k_2- 1}^0 x_{ \ell+ 1}$, and we can let
    $x_j'= x_j$ for $j\ne \ell$ and $x_\ell'= \delta_{ k_3}^1 x_{ \ell-
      1}$.

    For the remaining case $k_2= k_3$, we replace $x_{ \ell- 1}$ by
    another cube of equal dimension first: If $k_1< k_2$, then $x_{
      \ell- 2}= \delta_{ k_1}^0 \delta_{ k_2}^0 x_\ell= \delta_{ k_2-
      1}^0 \delta_{ k_1}^0 x_\ell$, hence the cube path $( x_1'',\dots,
    x_m'')$ with $x_j''= x_j$ for $j\ne \ell- 1$ and $x_{ \ell- 1}''=
    \delta_{ k_1}^0 x_\ell$ is adjacent to $( x_1,\dots, x_m)$, and $T(
    x_1'',\dots, x_m'')= T( x_1,\dots, x_m)$.  For this new cube path,
    we have $x_{ \ell- 2}''= \delta_{ k_2- 1}^0 x_{ \ell- 1}''$, $x_{
      \ell- 1}''= \delta_{ k_1}^0 x_\ell''$, and $x_{ \ell+ 1}''=
    \delta_{ k_3}^1 x_\ell''$, and as $k_1< k_3$, we can apply to the
    cube path $( x_1'',\dots, x_m'')$ the argument for the case $k_2<
    k_3$ above.

    If $k_1\ge k_2$, then $x_{ \ell- 2}= \delta_{ k_1}^0 \delta_{ k_2}^0
    x_\ell= \delta_{ k_2}^0 \delta_{ k_1+ 1}^0 x_\ell$ by another
    application of the precubical identity~\eqref{eq:pcub}.  Hence we
    can let $x_j''=x_j$ for $j\ne \ell- 1$ and $x_{ \ell- 1}''= \delta_{
      k_1+ 1}^0 x_\ell$.  Then $x_{ \ell- 2}''= \delta_{ k_2}^0 x_{
      \ell- 1}''$, $x_{ \ell- 1}''= \delta_{ k_1+ 1}^0 x_\ell''$, and
    $x_{ \ell+ 1}''= \delta_{ k_3}^1 x_\ell''$, and as $k_1+ 1> k_3$, we
    can apply the argument for the case $k_2> k_3$ above. 
  \end{proof}

  Now for the proof of Proposition~\ref{thm:univ_cov}, it is clear that
  the structure maps $\tilde \delta_k^1$ are well-defined.  For showing
  that also the mappings $\tilde \delta_k^0$ are well-defined, we note
  first that $\tilde \delta_k^0[ x_1,\dots,x_m]$ is independent of the
  representative chosen for $[ x_1,\dots, x_m]$: If $( x_1',\dots,
  x_m')\sim( x_1,\dots, x_m)$, then $( y_1,\dots, y_p)\in \tilde
  \delta_k^0[ x_1',\dots, x_m']$ if and only if $y_p= \delta_k^0 x_m'=
  \delta_k^0 x_m$ and $( y_1,\dots, y_p, x_m')=( y_1,\dots, y_p,
  x_m)\sim( x_1',\dots, x_m')\sim( x_1,\dots, x_m)$, if and only if $(
  y_1,\dots, y_p)\in \tilde \delta_k^0[ x_1,\dots, x_m]$.

  We are left with showing that $\tilde \delta_k^0[ x_1,\dots, x_m]$ is
  non-empty.  By Lemma~\ref{lem:fan} there is a fan-shaped cube path $(
  x_1',\dots, x_m')\in[ x_1,\dots, x_m]$, and by
  Lemma~\ref{lem:hom_fullcube} we can assume that $x_{ m- 1}'=
  \delta_k^0 x_m'= \delta_k^0 x_m$, hence $( x_1',\dots, x_{ m- 1}')\in
  \tilde \delta_k^0[ x_1,\dots, x_m]$.



  We need to show the precubical identity $\tilde \delta_k^\nu \tilde
  \delta_\ell^\mu= \tilde \delta_{ \ell- 1}^\mu \tilde \delta_k^\nu$ for
  $k< \ell$ and $\nu, \mu\in\{ 0, 1\}$.  For $\nu= \mu= 1$ this is
  clear, and for $\nu= \mu= 0$ one sees that $( y_1,\dots, y_p)\in
  \tilde \delta_k^0 \tilde \delta_\ell^0[ x_1,\dots, x_m]$ if and only
  if $y_p= \delta_k^0 \delta_\ell^0 x_m= \delta_{ \ell- 1}^0 \delta_k^0
  x_m$ and $( x_1,\dots, x_m)\sim ( y_1,\dots, y_p, \delta_\ell^0 x_m,
  x_m)\sim( y_1,\dots, y_p, \delta_k^0 x_m, x_m)$, by adjacency.

  The cases $\nu= 1$, $\mu= 0$ and $\nu= 0$, $\mu= 1$ are similar to
  each other, so we only show the former.  Let $( x_1',\dots, x_m')\in[
  x_1,\dots, x_m]$ be a fan-shaped cube path with $x_{ m- 1}'=
  \delta_\ell^0 x_m'$, \cf~Lemma~\ref{lem:hom_fullcube}.  Then $\tilde
  \delta_k^1 \tilde \delta_\ell^0[ x_1,\dots, x_m]= \tilde \delta_k^1[
  x_1',\dots, x_{ m- 1}']=[ x_1',\dots, x_{ m- 1}', \delta_k^1 x_{ m-
    1}']$.  Now $\delta_k^1 x_{ m- 1}'= \delta_k^1 \delta_\ell^0 x_m'=
  \delta_{ \ell- 1}^0 \delta_k^1 x_m$, and by adjacency, $( x_1',\dots,
  x_{ m- 1}', \delta_k^1 x_{ m- 1}', \delta_k^1 x_m')\sim( x_1',\dots,
  x_{ m- 1}', x_m', \delta_k^1 x_m')$, so that we have $( x_1',\dots,
  x_{ m- 1}', \delta_k^1 x_{ m- 1}')\in \tilde \delta_{ \ell- 1}^0[
  x_1',\dots, x_m', \delta_k^1 x_m']= \tilde \delta_{ \ell- 1}^0 \tilde
  \delta_k^1[ x_1',\dots, x_m']$.

  For showing that the projection $\pi_X: \tilde X\to X$ is a precubical
  morphism, we note first that $\pi_X \tilde \delta_k^1[ x_1,\dots,
  x_m]= \pi_X[ x_1,\dots, x_m, \delta_k^1 x_m]= \delta_k^1 x_m=
  \delta_k^1 \pi_X[ x_1,\dots, x_m]$ as required.  For $\tilde
  \delta_k^0$, let again $( x_1',\dots, x_m')\in[ x_1,\dots, x_m]$ be a
  fan-shaped cube path with $x_{ m- 1}'= \delta_k^0 x_m'$.  Then $\pi_X
  \tilde \delta_k^0[ x_1,\dots, x_m]= \pi_X[ x_1',\dots, x_{ m- 1}']=
  x_{ m- 1}'= \delta_k^0 x_m'= \delta_k^0 x_m= \delta_k^0 \pi_X[
  x_1,\dots, x_m]$.

  The proof that $*\to \tilde X$ is a higher-dimensional tree follows
  from Lemma~\ref{lem:cube_path_rep}: Let $( \tilde x_1,\dots, \tilde
  x_m)$, $( \tilde y_1,\dots, \tilde y_m)$ be pointed cube paths in
  $\tilde X$ with $\tilde x_m= \tilde y_m$, then we need to prove that
  $( \tilde x_1,\dots, \tilde x_m)\sim( \tilde y_1,\dots, \tilde y_m)$.
  Let $x_j= \pi_X \tilde x_j$, $y_j= \pi_X \tilde y_j$ for $j= 1,\dots,
  m$ be the projections, then $( x_1,\dots, x_m)$, $( y_1,\dots, y_m)$
  are pointed cube paths in $X$.  By Lemma~\ref{lem:cube_path_rep}, $(
  x_1,\dots, x_j)\in \tilde x_j$ and $( y_1,\dots, y_j)\in \tilde y_j$
  for all $j= 1,\dots, m$.

  By $\tilde x_m= \tilde y_m$, we know that $( x_1,\dots, x_m)\sim(
  y_1,\dots, y_m)$.  Let $( x_1,\dots, x_m)=( z^1_1,\dots,
  z^1_m)\sim\cdots\sim( z^p_1,\dots, z^p_m)=( y_1,\dots, y_m)$ be a
  sequence of adjacencies, and let $\tilde z^\ell_j=[ z^\ell_1,\dots,
  z^\ell_j]$.  This defines pointed cube paths $( \tilde z^\ell_1,\dots,
  \tilde z^\ell_m)$ in $\tilde X$; we show that $( \tilde x_1,\dots,
  \tilde x_m)=( \tilde z^1_1,\dots, \tilde z^1_m)\sim\cdots\sim( \tilde
  z^p_1,\dots, \tilde z^p_m)=( \tilde y_1,\dots, \tilde y_m)$ is a
  sequence of adjacencies:

  Let $\ell\in\{ 1,\dots,p- 1\}$, and let $\alpha\in\{ 1,\dots, m- 1\}$
  be the index such that $z^\ell_\alpha\ne z^{ \ell+ 1}_\alpha$ and
  $z^\ell_j= z^{ \ell+ 1}_j$ for all $j\ne \alpha$.  Then $(
  z^\ell_1,\dots, z^\ell_j)=( z^{ \ell+1}_1,\dots, z^{ \ell+ 1}_j)$ for
  $j< \alpha$ and $( z^\ell_1,\dots, z^\ell_j)\sim( z^{ \ell+1}_1,\dots,
  z^{ \ell+ 1}_j)$ for $j> \alpha$, hence there is an adjacency $(
  \tilde z^\ell_1,\dots, \tilde z^\ell_m)\sim( \tilde z^{ \ell+
    1}_1,\dots, \tilde z^{ \ell+ 1}_m)$.

  For the last claim of the proposition, if $X$ itself is a
  higher-dimensional tree, then an inverse to $\pi_X$ is given by
  mapping $x\in X$ to the unique equivalence class $[ x_1,\dots, x_m]\in
  \tilde X$ of any pointed cube path $( x_1,\dots, x_m)$ in $X$ with
  $x_m= x$.
\end{proof*}

\begin{proof*}{Proof of Lemma~\ref{lem:cube_path_rep}.}
  Let $x_j= \pi_X \tilde x_j$, for $j= 1,\dots, m$, then $( x_1,\dots,
  x_m)$ is a pointed cube path in $X$.  We show the claim by induction:
  We have $\tilde x_1= \tilde i=[ i]=[ x_1]$, so assume that $(
  x_1,\dots, x_j)\in \tilde x_j$ for some $j\in\{ 1,\dots, m- 1\}$.  If
  $\tilde x_{ j+ 1}= \tilde \delta_k^1 \tilde x_j$ for some $k$, then
  $x_{ j+ 1}= \delta_k^1 x_j$, and $( x_1,\dots, x_{ j+ 1})\in \tilde
  x_{ j+ 1}$ by definition of $\tilde \delta_k^1$.  Similarly, if
  $\tilde x_j= \tilde \delta_k^0 \tilde x_{ j+ 1}$ for some $k$, then
  $x_j= \delta_k^0 x_{j+ 1}$, and $( x_1,\dots, x_{ j+ 1})\in \tilde x_{
    j+ 1}$ by definition of $\tilde \delta_k^0$.
\end{proof*}

\begin{proof*}{Proof of Lemma~\ref{le:unf_unique_lift}.}
  Let $( \tilde x_1,\dots, \tilde x_m)$ be a pointed cube path in
  $\tilde X$, and write $x_j= \pi_X \tilde x_j$ for $j= 1,\dots, m$.
  Let $( x_1,\dots,x_m, y_1,\dots, y_p)$ be an extension in $X$ and
  define $\tilde y_j=[ x_1,\dots, x_m, y_1,\dots, y_j]$ for $j= 1,\dots,
  p$.  Then $( \tilde x_1,\dots, \tilde x_m, \tilde y_1,\dots, \tilde
  y_p)$ is the required extension in $\tilde X$, which is unique as
  $\tilde X$ is a higher-dimensional tree.
\end{proof*}

\begin{proof*}{Proof of Lemma~\ref{lem:hdt_vs_hdth}.}
  Using the projection isomorphisms, any morphism $f: X\to Y$ in $\HDTh$
  can be ``pulled down'' to a morphism $\pi_Y\circ f\circ \pi_X^{ -1}:
  X\to Y$ of $\HDT$.
\end{proof*}

\begin{proof*}{Proof of Lemma~\ref{lem:openh=open}.}
  For the forward implication of the first claim, let
  \begin{equation}
    \label{eq:small-diag}
    \vcenter{%
      \xymatrix{%
        P \ar[r]^p \ar[d]_g & \tilde X \ar[d]^f \\ Q \ar[r]_q & \tilde Y
      }}
  \end{equation}
  be a diagram in $\HDA$ with $g: P\to Q\in \HDP$; we need to find a
  lift $Q\to \tilde X$.

  Using the isomorphisms $\pi_P: \tilde P\to P$, $\pi_Q: \tilde Q\to Q$,
  we can extend this diagram to the left; note that $\tilde g: \tilde
  P\to \tilde Q$ is a morphism of $\HDP$:
  \begin{equation}
    \label{eq:wide_diag}
    \vcenter{%
      \xymatrix{%
        \tilde P \ar[r]_{ \cong} \ar[d]_{ \tilde g} \ar@/^3ex/[rr]^{ p'} & P
        \ar[r]_p \ar[d]_g & \tilde X \ar[d]^f \\ \tilde Q \ar[r]^{ \cong}
        \ar@/_3ex/[rr]_{ q'} & Q \ar[r]^q & \tilde Y
      }}
  \end{equation}

  Hence we have a diagram
  \begin{equation*}
    \xymatrix{%
      P \ar[r]^{ p'} \ar[d]_{ \tilde g} & X \ar[d]^f \\ Q \ar[r]_{ q'} &
      Y
    }
  \end{equation*}
  in $\HDAh$, and as $\tilde g: P\to Q$ is a morphism of $\HDPh$, we
  have a lift $r: Q\to X$ in $\HDAh$.  This gives a morphism $r: \tilde
  Q\to \tilde X\in \HDA$ in Diagram~\eqref{eq:wide_diag}, and by
  composition with the inverse of the isomorphism $\pi_Q: \tilde Q\to
  Q$, a lift $r': Q\to \tilde X\in \HDA$ in
  Diagram~\eqref{eq:small-diag}.

  For the back implication in the first claim, assume $f: \tilde X\to
  \tilde Y\in \HDA$ open and let
  \begin{equation*}
    \xymatrix{%
      P \ar[r]^p \ar[d]_g & X \ar[d]^f \\ Q \ar[r]_q & Y
    }
  \end{equation*}
  be a diagram in $\HDAh$ with $g: P\to Q\in \HDPh$; we need to find a
  lift $Q\to X$.  Transferring this diagram to the category $\HDA$, we
  have
  \begin{equation*}
    \xymatrix{%
      \tilde P \ar[r]^p \ar[d]_g & \tilde X \ar[d]^f \\ \tilde Q
      \ar[r]_q & \tilde Y
    }
  \end{equation*}
  and as $g: \tilde P\to \tilde Q$ is a morphism of $\HDP$, we get the
  required lift.

  To prove the second claim, let
  \begin{equation*}
    \xymatrix{%
      P \ar[r]^p \ar[d]_h & \tilde X \ar[d]^(.6){ \tilde g} \\ Q \ar[r]_q &
      \tilde Y
    }
  \end{equation*}
  be a diagram in $\HDA$ with $h: P\to Q\in \HDP$.  We can extend it
  using the projection morphisms:
  \begin{equation*}
    \xymatrix{%
      P \ar[r]^p \ar[d]_h & \tilde X \ar[r]^{ \pi_X} \ar[d]^(.6){ \tilde g} &
      X \ar[d]^g \\ Q \ar[r]_q & \tilde Y \ar[r]_{ \pi_Y} & Y
    }
  \end{equation*}
  Because $g$ is open in $\HDA$, we hence have a lift
  \begin{equation*}
    \xymatrix{%
      P \ar[r]^p \ar[d]_h & \tilde X \ar[r]^{ \pi_X} \ar[d]^(.6){ \tilde g} &
      X \ar[d]^g \\ Q \ar[r]_q \ar@{.>}[rru]_(.7)r & \tilde Y \ar[r]_{ \pi_Y}
      & Y
    }
  \end{equation*}
  and Lemma~\ref{le:unf_unique_lift} then gives the required lift $r'$
  in the diagram
  \begin{equation*}
    \xymatrix{%
      P \ar[r]^p \ar[d]_g & \tilde X \ar[d]^{ \pi_X} \\ Q \ar[r]_r
      \ar@{.>}[ur]^{r'} & X
    }
  \end{equation*}
\end{proof*}

\begin{proof*}{Proof of Lemma~\ref{lem:reach-vs-prefix}.}
  For the forward implication, let $( \tilde x, \tilde y_1,\dots, \tilde
  y_p)$ be a cube path in $\tilde X$ with $\tilde y_p= \tilde z$, let $(
  x_1,\dots, x_m)\in \tilde x$, and write $y_j= \pi_X \tilde y_j$ for
  all $j$.  By Lemma~\ref{lem:cube_path_rep}, $( x_1,\dots, x_m,
  y_1,\dots, y_p)\in \tilde z$.

  For the other direction, let $( x_1,\dots, x_m, y_1,\dots, y_p)\in
  \tilde z$ such that $( x_1,\dots, x_m)\in \tilde x$, and define
  $\tilde y_j=[ x_1,\dots, x_m, y_1,\dots, y_j]$ for all $j$.  Then $(
  \tilde x, \tilde y_1,\dots, \tilde y_p)$ is the required cube path
  from $\tilde x$ to $\tilde z$ in $\tilde X$.
\end{proof*}

\begin{proof*}{Proof of Proposition~\ref{thm:hombis}.}
  The implication
  \eqref{enu:hombis.box}~$\Longrightarrow$~\eqref{enu:hombis.onestep}
  follows directly from Theorem~\ref{thm:bisim},
  and~\eqref{enu:hombis.cubepath} can be proven
  from~\eqref{enu:hombis.onestep} by induction.  (We can omit the
  reachability condition from items~\eqref{enu:hombis.onestep}
  and~\eqref{enu:hombis.cubepath} because any cube in an unfolding is
  reachable.)  Equivalence of~\eqref{enu:hombis.cubepath}
  and~\eqref{enu:hombis.prefix} is immediate from
  Lemma~\ref{lem:reach-vs-prefix}.

  For the implication
  \eqref{enu:hombis.cubepath}~$\Longrightarrow$~\eqref{enu:hombis.box},
  we can use Theorem~\ref{thm:bisim} to get a span $\tilde X\tfrom f
  R\tto g \tilde Y$ of open maps in $\HDA$.  Connecting these with the
  projection $\pi_R: \tilde R\to R$ gives a span $\tilde X\tfrom{ f\circ
    \pi_R} \tilde R\tto{ g\circ \pi_R} \tilde Y$.  By
  Corollary~\ref{cor:proj_open}, the maps in the span are open in
  $\HDA$, hence by Lemma~\ref{lem:openh=open}, $X\tfrom{ f\circ \pi_R}
  R\tto{ g\circ \pi_R} Y$ is a span of open maps in $\HDAh$.
\end{proof*}

\begin{proof*}{Proof of Theorem~\ref{th:hbis=bis}.}
  A span of open maps $X\tfrom f Z\tto g Y$ in $\HDA$ lifts to a span
  $X\tfrom{ \tilde f} Z\tto{ \tilde g} Y$ in $\HDAh$, and $\tilde f$ and
  $\tilde g$ are open by Lemma~\ref{le:liftopen}.  Hence bisimilarity
  implies homotopy bisimilarity.

  For the other direction, let $X\tfrom f Z\tto g Y$ be a span of open
  maps in $\HDAh$.  In $\HDA$, this is a span $\tilde X\tfrom f \tilde
  Z\tto g \tilde Y$, and composing with the projections yields $X\tfrom{
    \pi_X\circ f} \tilde Z\tto{ \pi_Y\circ g} Y$.  By
  Lemma~\ref{lem:openh=open} and Corollary~\ref{cor:proj_open}, both
  $\pi_x\circ f$ and $\pi_Y\circ g$ are open in $\HDA$.
\end{proof*}

\begin{proof*}{Proof of Theorem~\ref{thm:hombis=histpres}.}
  For the ``if'' part of the theorem, assume that we have a relation $R$
  as in Definition~\ref{defi:histpres} and define $\tilde R\subseteq
  \tilde X\times \tilde Y$ by $\tilde R=\{( \tilde x, \tilde y)\mid
  \exists \rho\in \tilde x, \sigma\in \tilde y:( \rho, \sigma)\in R\}$.
  Then $( \tilde i, \tilde j)\in \tilde R$.  Now let $( \tilde x_1,
  \tilde y_1)\in \tilde R$ and $\tilde x_2\sqsupseteq \tilde x_1$.  We
  have $\rho_1\in \tilde x_1$ and $\sigma_1\in \tilde y_1$ for which $(
  \rho_1, \sigma_1)\in R$.  Let $\rho_1'\in \tilde x_1$ and $\rho_2\in
  \tilde x_2$ such that $\rho_2\sqsupseteq \rho_1'$, then $\rho_1'\sim
  \rho_1$, hence we have $\sigma_1'\sim \sigma_1$ for which $( \rho_1',
  \sigma_1')\in R$.  By $\rho_2\sqsupseteq \rho_1'$ we also have
  $\sigma_2\sqsupseteq \sigma_1'$ for which $( \rho_2, \sigma_2)\in R$,
  hence $( \tilde x_2=[ \rho_2],[ \sigma_2])\in \tilde R$ as was to be
  shown.  The symmetric condition in
  Theorem~\ref{thm:hombis}\eqref{enu:hombis.prefix} can be shown
  analogously, and one easily sees that $\tilde R$ is indeed a
  precubical set.

  For the other implication, let $\tilde R\subseteq \tilde X\times
  \tilde Y$ be a precubical set as in
  Theorem~\ref{thm:hombis}\eqref{enu:hombis.prefix} and define a
  relation of pointed cube paths by $R=\{( \rho, \sigma)\mid([ \rho],[
  \sigma])\in \tilde R\}$.  Then $(( i),( j))\in R$.  Now let $( \rho,
  \sigma)\in R$, then also $( \rho', \sigma')\in R$ for any $\rho'\sim
  \rho$, $\sigma'\sim \sigma$, showing the first two conditions of
  Definition~\ref{defi:histpres}.  For the third one, let
  $\rho'\sqsupseteq \rho$, then $[ \rho']\sqsupseteq[ \rho]$, hence we
  have $\tilde y_2\sqsupseteq[ \sigma]$ for which $([ \rho'], \tilde
  y_2)\in \tilde R$.  By definition of $R$ we have $( \rho', \sigma')\in
  R$ for any $\sigma'\in \tilde y_2$, and by $\tilde y_2\sqsupseteq[
  \sigma]$, there is $\sigma'\in \tilde y_2$ for which
  $\sigma'\sqsupseteq \sigma$, showing the third condition.  The fourth
  condition is proved analogously.
\end{proof*}

\begin{proof*}{Proof of Corollary~\ref{co:decidable}.}
  The condition in Thm.~\ref{thm:bisim}\eqref{enu:bisim.onestep}
  immediately gives rise to a fixed-point algorithm similar to the one
  used to decide standard bisimilarity, \cf~\cite{book/Milner89}.
\end{proof*}

\begin{proof*}{Proof of Theorem~\ref{th:hombis=bistpres-l}.}
  The proof is similar to the one of Theorem~\ref{thm:hombis=histpres}.
  For the ``if'' part, the condition $\lambda( \rho)= \mu( \sigma)$
  ensures that the homotopy bisimilarity relation respects homotopy
  classes of split traces, and for the ``only if'' part, starting with a
  homotopy bisimilarity relation $\tilde R\subseteq \tilde X\times
  \tilde Y$, we have to define the history-preserving bisimilarity
  relation $R$ by $R=\{( \rho, \sigma)\mid ([ \rho],[ \sigma])\in \tilde
  R, \lambda( \rho)= \mu( \sigma)\}$ instead.
\end{proof*}

\end{document}